\RequirePackage{thm-restate}
\documentclass[a4paper,USenglish,cleveref,autoref,thm-restate]{article}
\usepackage{tda2}

\begin{document}

\title{How Small Can Faithful Sets Be? Ordering Topological Descriptors}

\author{}
%

\date{}

\author{Brittany Terese Fasy
\thanks{School of Computing and Dept.\ Mathematical Sciences,
            Montana State
	    U., \tt brittany.fasy@montana.edu}
	\and
	David L. Millman \thanks{Blocky Inc., \tt david@blocky.rocks}
        \and
        Anna Schenfisch \thanks{Dept.\ Mathematics and Computer Science,
        Eindhoven U. of Technology, \tt a.k.schenfisch@tue.nl}}

\index{Anna Schenfisch}
\index{Brittany Terese Fasy}

\maketitle

\begin{abstract}
    Recent developments in shape reconstruction and comparison call for the use of
    many different (topological) descriptor types, such as persistence
    diagrams and Euler characteristic functions. We establish a framework to
    quantitatively compare the strength of different descriptor types, setting
    up a theory that
    allows for future comparisons and analysis of descriptor types and that can inform
    choices made in applications.  We
    use this framework to partially order a set of six common descriptor types.
    We then give lower bounds on
    the size of sets of descriptors that uniquely correspond to simplicial
    complexes, giving insight into the advantages of using verbose
    rather than concise topological~descriptors.
\end{abstract}

\section{Introduction}\label{sec:intro}
The persistent homology transform and Euler characteristic transform
were first explored in~\cite{turner2014persistent}, which
shows the uncountable set of persistence diagrams (or Euler characteristic
functions, respectively) corresponding to lower-star filtrations in every
possible direction uniquely represents the shape being filtered. That is, the
uncountable set of topological descriptors is \emph{faithful} for the shape.
Faithfulness of topological transforms is closely related to \emph{tomography}
\cite{schapira1995tomography, lebovici2022hybrid}, and the alternate proof of
faithfulness
given in~\cite{ghrist2018euler} makes use of tools from this field.
Of course, applications can only use finite sets of descriptors, which
are not guaranteed to be faithful.  This motivates theoretical work on
finding finite faithful sets of descriptors~\cite{belton2019reconstructing,
curry2022many,fasy2022efficient, mickaPhD}, and such work
supports the use of topological descriptors in shape comparison applications. 
Many descriptor types are used in applications, such as versions
of persistence diagrams
\cite{lee2017quantifying, rizvi2017single, lawson2019persistent,
tymochko2020using, wang2019statistical,bendich2016persistent},
Euler characteristic functions~\cite{jiang2020weighted, amezquita2022measuring,
nadimpalli2023euler, richardson2014efficient,crawford2020predicting,
meng2022randomness,marsh2023stability}, Betti
functions~\cite{saadat2021topological, li2023classification,
frosini2013persistent,pranav2017topology, wilding2021persistent,van2010alpha},
and others~\cite{giusti2015clique,
singh2007topological,batakci2023comparing}.

Faithfully representing a shape with a small number of descriptors is desirable
for computational and storage reasons.  How, then, should investigators choose
the particular topological descriptor type to use in applications? While
computational complexities of computing each topological descriptor type are
well-studied, it is not yet known how the use of particular descriptor types
impacts the minimum size of faithful sets.  This uncertainty motivates our main
questions: \emph{how can we rigorously compare descriptor types in terms of
their ability to uniquely correspond to shapes}, and \emph{how do popular
descriptor types compare?}

We prove the partial order of \figref{abstract}.
\begin{figure}[h!]
    \centering
    \includegraphics[width=.8\textwidth]{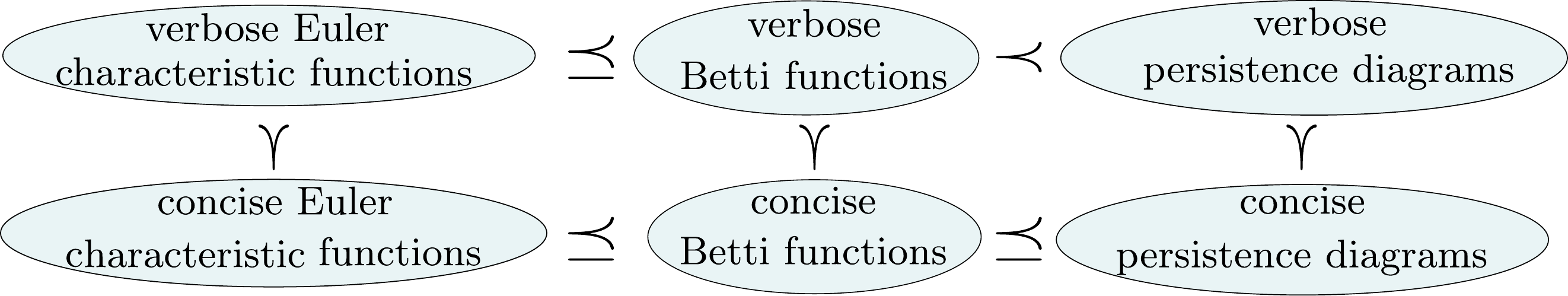}
	\caption{Summary of the relations between common descriptor types.
        For example, concise persistence diagrams are always at
	    least as efficient as concise Betti functions at forming faithful sets,
	    where efficiency is measured by cardinality of faithful sets.
           }
\label{fig:abstract}
\end{figure}
Additionally, we provide lower bounds on the cardinality of faithful sets for both
concise and verbose descriptors, and identify properties that
indicate concise descriptors are
generally much weaker than verbose descriptors. This suggests applications
research may benefit from the use of verbose descriptors instead of the more
widely adopted concise descriptors.

\section{Preliminary Considerations}\label{sec:preliminary}
In this section, we provide background and definitions used throughout.

\paragraph{Simplicial Complexes and Filtrations}
We assume the reader is familiar with foundational ideas from topology, such as homology, Betti number ($\beta_k$) and Euler characteristic $(\chi)$.
See, e.g.,~\cite{hatcher, edelsbrunner2010computational}.
For a simplicial complex~$K$ and $i \in \N$, we use the notation $K_i$ for
the set of its~$i$-simplices
and~$n_i$ as the number of~$i$-simplices.
Furthermore, we assume our simplicial complexes are
\emph{abstract} simplicial complexes immersed in
Euclidean space such that each simplex is embedded
and the vertices are in general position. Specifically, this is:
\begin{assumption}
    \label{ass:mostgeneral}
    A simplicial complex $K$ immersed in $\R^d$ is in
    \emph{general position} if, for all~$V \subseteq K_0$ with $|V|\leq
    d+1$, the set $V$ is affinely independent
\end{assumption}

A \emph{filter} of~$K$ is a monotone
map~$f \colon K \to \R$ such that, each sublevel set~$f^{-1}(-\infty, {t}]$ is
either empty or a simplicial complex.  Letting~$F(t) :=
f^{-1}(-\infty, {t}]$, the sequence~$\{ F(t)  \}_{t \in \R}$ is the
\emph{filtration} associated to~$f$. For each~$k \in \N$, the inclusion~$F(i)
\hookrightarrow F(j)$ induces a linear map on homology,~$H_k(F(i)) \to
H_k(F(j))$.  We write~$\beta^{i,j}_k(K,f)$ to mean rank of this map, or
simply~$\beta_k^{i,j}$ if $K$ and $f$ are clear from context.
We call a filter function~$f': K \to \{ 1,2, \ldots, \# K\}$ a \emph{compatible
index filter} for $f$ if, for all~$\tau, \sigma \in K$ with $f(\tau) \leq
f(\sigma)$, then we have~$f'(\tau) \leq f'(\sigma)$. Every filter function has at least
one compatible index filter.

The \emph{lower-star filter} of a simplicial complex $K$ immersed in
$\R^d$ with respect to some direction~$s
\in \S^{d-1}$, is the map~$f_s: K \to \R$ that takes a simplex $\sigma$ to the
maximum height of its vertices with respect to direction~$s$, i.e.,~$f_s(\sigma) :=
\max\{s\cdot v \mid v \in K_0 \cap \sigma\}$, where~$s \cdot v$ denotes the dot
product.

\paragraph{Faithfully Representing a Simplicial Complex}
Since we define relations based on the ability of descriptor types to represent
particular filtrations of simplicial complexes, we take the following
definition.

\begin{definition}[Topological Descriptors]
    \label{def:descriptors}
    A \emph{(topological) descriptor type} is a map whose domain is the
	collection of filtered simplicial complexes. Given
        such map, $\gen$, a \emph{(topological)
	descriptor of type~$D$} is the image of a specific
	filtered simplicial complex under~$D$.
\end{definition}

When considering many filtrations of the same simplicial complex, we
may index the filtrations by some parameter set, $P$. If a descriptor of type
$\gen$ corresponds to a filtration of a simplicial complex $K$ where the
filtration is parameterized by~$p \in P$, we use the notation
$\dirDescLong{\gen}{\simComp}{p}$, or $\dirDesc{\gen}{p}$ when $K$ is clear
from context. We refer to the parameterized set of descriptors
as~$\descSet{\gen}{K}{P} := \{\left(p, \gen(K,p) \right)\}_{p \in P}$.

We compare descriptor types by their ability to efficiently and
uniquely identify a shape. The ability for a set of descriptors to 
uniquely identify a shape is formalized as follows: 
\begin{definition}[Faithful]\label{def:faithful}
	Let $K$ be a simplicial complex, $P$
	parameterize a set of
	filtrations of~$K$, and $\gen$ be a
	topological descriptor type. We say that~\emph{$\descSet{\gen}{K}{P}$
	is faithful} if, for any simplicial complex~$L$ we
	have the equality~$\descSet{\gen}{L}{P} = \descSet{\gen}{K}{P}$ if and
	only if~$L= K$.
\end{definition}

To unpack the equality $\descSet{\gen}{L}{P} =
\descSet{\gen}{K}{P}$, recall
$
    \descSet{\gen}{K}{P} := \{\left(p, \gen(K,p) \right)\}_{p \in P}.
$
Thus, $\descSet{\gen}{L}{P} =
\descSet{\gen}{K}{P}$ if and only if for all $p \in P$, we have
$\gen(K,p)=\gen(L,p)$. Then, $\gen(K,P)$ is faithful if and only if 
\begin{equation}
    \bigcap_{p \in P} \left \{K' \subset \R^d \mid \gen(K', p) = \gen(K,
    p)\right\}
    = \{K\}.
\end{equation}
From this perspective, we prove the following lemma providing a sufficient
condition for \emph{finite} faithful sets.

\begin{lemma}[Sufficient Conditions for Finite Faithful Set Existence]\label{lem:finiteset}
    Let $\simComp$ be a simplicial complex immersed in~$\R^d$ and let $\gen$ be a
    type of topological descriptor that can faithfully represent~$K$. Suppose
    there exists a finite set of descriptors of type $\gen$ that is faithful
    for~$K_0$. Then, there exists a finite faithful set of descriptors of type
    $\gen$ that is faithful for~$K$.
\end{lemma}
\begin{proof}
    Let $P$ be a parameter set such that  $\descSet{\gen}{K}{P}$ that is
    faithful for~$K$, and let $P_0$ be a finite parameter set such that
    $\descSet{\gen}{K}{P_0}$ that is faithful for $K_0$.
    Let $B$ be the set of simplicial complexes that are indistinguishable from
    $K$ using only parameter set~$P_0$; that is,
    \begin{equation}
        B := \bigcap_{p \in P_0} \left \{K' \subset \R^d \mid \gen(K', p) = \gen(K,
        p)\right\}.
    \end{equation}
    Since~$P_0$ is faithful for $K_0$, we know that $B \subseteq \{K' \mid K'_0 = K_0\}$,
    i.e., $B$ is a subset of all simplicial complexes built out of the vertices of
    $K$. In particular, we note that this set is finite; since $n_0$ is
    finite, there are a finite number of simplicial complexes we can build over
    this set of vertices.

    If $B=\{ K\}$, we are done.  Otherwise,
    since $D(K,P)$ is faithful for $K$,
    for each~$L\neq K$ in~$B$, there exists some~$p_L
    \in P$
    such that $\gen(L, p_L) \neq \gen(K, p_L)$.
    Let~$P^*=P \cup \{ p_L \mid L \in B \}$.
    Then,
    $\descSet{\gen}{K}{P^*}$ faithfully represents~$K$.
    Furthermore, since $P$ and $B$ are finite, we
    also know that $P^*$ is finite.
\end{proof}

\section{Six Common Descriptor Types}\label{sec:descriptors}
\label{ss:standard}
The set we partially order is the strength equivalence
classes of six popular descriptor types, which we define here.  We
begin with concise persistence diagrams.

\begin{definition}[Concise Persistence Diagram]
    Let~$f: \simComp \to \R$ be a filter function.
    For $k \in \N$, the~\emph{$k$-dimensional concise persistence
    diagram} is:
    \begin{align*}
        \pd_k^f:=  \big\{ &(i,j)^{\mu^{(i,j)}}  \text{ s.t. }
        (i,j) \in \overline{\R}^2 \\
        &\text{ and }
        \mu^{(i,j)}=\beta^{i,j-1}_k - \beta^{i,j}_k - \beta^{i-1,j-1}_k +
        \beta^{i-1,j}_k \big\},
    \end{align*}
    where~$\overline{\R} = \R \cup \{\pm \infty \}$ and~$(i,j)^m$ denotes $m$
    copies of the point~$(i,j)$.
    The \emph{concise persistence diagram} of $f$, denoted $\pd^f$, is the
    indexed union of all $k$-dimensional concise persistence diagrams~$\pd^f:=
	\cup_{k \in \N} \pd_{k}^{f}$.
\end{definition}
Since simplices can appear at the same parameter value in a filtration,
not all cycles are represented in the concise persistence diagram.  However, having
every simplex ``appear'' in a topological descriptor is helpful, in addition to
being natural.  Thus, we introduce \emph{verbose} descriptors, which contain
this information. We define verbose descriptors via compatible index
filtrations; by Lemma 52 and Corollaries 54-55
of~\cite{fasy2019reconstructing}, this is well-defined and
independent of
our choice of compatible index filtration. We begin with \emph{verbose persistence diagrams}:

\begin{definition}[Verbose Persistence Diagram]\label{def:apd}
    Let $f: \simComp \to \R$ be a filter for $K$, and let $f'$ be a
    compatible index filter. For $k \in \N$, the \emph{$k$-dimensional
    verbose persistence diagram} is the following multiset:
    \begin{equation*}\label{eqn:apd-pts}
        \apd_{k}^{f} := \left\{ \left(f(\sigma_i), f(\sigma_j) \right)
        \text{ s.t. } {(i,j) \in \pd_{k}^{f'}} \right\}.
    \end{equation*}
    The \emph{verbose persistence diagram} of $f$, denoted~$\apd^f$, is the
    indexed union of all $\apd_{k}^{f}$.
\end{definition}

Recording invariants other than homology leads to other
topological descriptor types; recording Betti numbers gives us concise or
verbose
\emph{Betti functions}.

\begin{definition}[Betti Functions]
    Let~$f \colon \simComp \to \R$ be a filter function.
    The~\emph{$k$th concise Betti function},
    $\bc_k^f: \R \to \Z$, is
    defined by
    \begin{equation*}\label{eqn:bc}
        \bc_k^f(t) := \beta_k \left( f^{-1}(-\infty,t]\right).
    \end{equation*}
    The indexed collection of such functions for all dimensions,~$\bc^f:= \{
        \bc_k^f ~|~ k \in \N \}$, is the \emph{concise Betti function}.

    Let $f'$ be an index
    filter compatible with filter function $f$.
    We call~$\sigma \in \simComp$ \emph{positive} (respectively,
    \emph{negative}) \emph{for~$\beta_k$} if the inclusion of $\sigma$ into the index
    filtration of $f'$ increases (resp., decreases)
    $\beta_k$. We denote the positive (resp., negative) simplices by $K_k^+
    \subseteq K_k$
    (and~$K_{k+1}^- \subseteq K_{k+1}$).
    Then, the~\emph{$k$th verbose
    Betti function},~$\abc_k^{f}: \R \to \Z^2$, is defined by
    \begin{align*}
        \label{eqn:abc-pts}
        \abc_k^{f}(p)
        := \bigg(&\left|\{\sigma \in K_k^+  \text{ s.t. } f(\sigma) \leq p\}\right|,\\
                 & \left|\{\sigma \in K_{k+1}^- \text{ s.t. } f(\sigma) \leq p\}\right|
           \bigg).
    \end{align*}
    The collection of verbose Betti number
    functions for each dimension is known as the \emph{verbose Betti function}
    and is denoted~$\abc^f$.
\end{definition}

If we record \emph{Euler characteristic} in a filtration, we obtain concise or
verbose \emph{Euler characteristic functions}.\footnote{Euler characteristic
functions and Betti functions are sometimes called Euler (characteristic) curves
or Betti curves.}

\begin{definition}[Euler Characteristic Functions]\label{def:aecc}
    Let $f \colon \simComp \to \R$ be a filter function.
    The \emph{concise Euler characteristic
    function},
    $\stecc^f: \R \to \Z$, is
    defined by:
    \begin{equation*}\label{eqn:ecc}
	    \steccdec{p}{f} := \chi \left( \{f^{-1}(-\infty,p]\} \right).
    \end{equation*}

    Let $f'$ be an index filter compatible with $f$.
    We call~$\sigma \in \simComp$ \emph{even} (respectively,
    \emph{odd}) if the dimension of $\sigma$ is even (resp., odd).
    Denoting the set of
    even (resp., odd) simplices by~$E$ (and~$O$),
    the \emph{verbose Euler
    characteristic function},~$\augecc^f: \R \to \Z^2$, is
    defined~by
    \begin{equation*}\label{eqn:aecc-pts}
        \augeccdec{p}{f} := \left(
                | \sigma \in E  \text{ s.t. } f(\sigma) \leq p |,
                | \sigma \in O  \text{ s.t. } f(\sigma) \leq p |
    	    \right).
    \end{equation*}
    In other words, $\aecc^f$ represents $\ecc^f$ as a parameterized count of
    even- and odd-dimensional simplices.
\end{definition}
In each of the descriptor types above, we drop the superscript $f$ when it is
clear from context.
See \appendref{six} for examples of these descriptor types.
While concise descriptors may feel more familiar, verbose descriptors are not
new.  Many algorithms for computing persistence (e.g.,~\cite[Chapter
VII]{edelsbrunner2010computational}), explicitly compute events with trivial
lifespan.  In~\cite{PHFunc}, the
definition of persistence diagrams agrees with our \defref{apd}.  Verbose
descriptors are closely connected to the charge-preserving morphisms
of~\cite{fasy2022persistent,mccleary2022edit}.  In~\cite{usher2016persistent},
verbose persistence is defined via filtered chain
complexes;~\cite{chacholski2023decomposing, memoli2022stability,
memoli2023ephemeral,zhou2023beyond} also take this view as a foundational
definition. The behavior of verbose versus concise descriptors is
explored in~\cite{fasy2018euler, memoli2023ephemeral,zhou2017exploring}.

Verbose (concise) descriptors are sometimes called augmented (non-augmented,
respectively) in the literature. We refer to points on a verbose
diagram with zero-lifespan as
\emph{instantaneous}. Such points correspond to length-zero barcodes in a
verbose barcode, which are sometimes referred to as \emph{ephemeral}.

While we chose the six descriptor types above due to
their relevance in applications, we emphasize that \defref{descriptors}
is very general. We explore
a few pathological descriptor types in \appendref{other-descriptors}.

\camera{BTF: there is a bit of a disconnect between the descriptors being really
general and the equiv classes being for representing PHT / parameterized by sets
of directions. we have somewhere discussed the discretization of the PHT. it
might be worth to discuss things like that.  but, maybe that should be for the
journal version}

\section{Relating Descriptor Types}\label{sec:relations}
We now develop tools to compare descriptor types, by
comparing the sizes of faithful sets.
We note here that~$|\descSet{\gen}{K}{P}| = |P|$, and for brevity of notation,
we elect to write the later.
Given a topological descriptor type~$D$ and simplicial complex $K$ immersed in
$\R^d$, we denote the infimum size of faithful sets for $K$ as
\begin{equation*}
    \Gamma(K,D):= \inf_{D(K,P) \text{ faithful}}
        \left\{ \lvert  P \rvert \right\}.
\end{equation*}
Intuitively, the stronger~$D$ is, the smaller~$\Gamma(K,D)$.  Often, we
find~$\Gamma(K,D)$ is finite.  For some descriptors and~$K$, we
find~$\Gamma(K,D)=\aleph_0$ (the cardinality of~$\N$) or~$\Gamma(K,D)=\aleph_1$
(the cardinality of $\R$); see \appendref{other-descriptors} for examples.  If
no faithful set of type $D$ exists for~$K$, we write~$\inf_{x \in \emptyset} \{
    x \} = \aleph_{\top}$, and we think of this as ``the highest''
cardinality.\footnote{We hope any discomfort caused by our use of
$\aleph_{\top}$ will be outweighed by the benefit of being able to avoid
lengthy and awkward case analyses.}
By the axiom of choice, $\aleph_0 < \aleph_1$; see
e.g.,~\cite[Ch.~2]{jech1981set}.  Thus, we have a total order on possible
values of $\Gamma(K,D)$:
\begin{equation*}
    c < \aleph_0 < \aleph_1 < \aleph_{\top},
\end{equation*}
where $c \in \N$.

\begin{definition}[Strength Relation]
    \label{def:equal}\label{def:preceq}
    Let $A$ and $B$ be two
    topological descriptor types. If, for every simplicial
    complex $K$ immersed in $\R^d$, we have~$\Gamma(K,A)\geq \Gamma(K,B)$, then we say that $A$ is
    \emph{weaker} than $B$ (and~$B$ is \emph{stronger} than $A$)
    denoted $\st{A} \preceq \st{B}$.
    If $\st{A} \preceq \st{B}$ and $\st{B} \preceq \st{A}$, then we say that $A$
    and $B$ have equal strength, denoted~$[A]=[B]$.
\end{definition}

The relations $=$ are $\preceq$ are well-defined on strength equivalence
classes. 

\begin{lemma}\label{lem:well-defined} \label{lem:preceqwelldef}
    The relation $=$ is an equivalence relation, and
    the relation $\preceq$ is well-defined on sets of strength equivalence classes.
\end{lemma}

\begin{proof}
    When we compare infimums in \defref{preceq}, we compare values in $\N
    \cup \{\aleph_0, \aleph_1, \aleph_{\top}\}$.  The
    relation~$\leq$ on
    values in this set is reflexive, antisymmetric, and transitive. The
    relation $=$ on this set is reflexive, symmetric, and
    transitive. The result follows.
\end{proof}

See \exref{sameclass} of \appendref{other-descriptors} for two
different descriptor types in the same equivalence class.

We write $\st{A} \prec \st{B}$ if $\st{A}
\preceq \st{B}$ and~$\st{A}
\neq \st{B}$.
That is, if $\st{A} \preceq \st{B}$ and there
exists a simplicial complex for which the minimum faithful set of
type $B$ is strictly smaller than that of type $A$, or for
    which there exists a faithful set of type $B$ but not of type $A$.
Descriptor types need not be comparable;
see \lemref{incomparable} of \appendref{other-descriptors}.

We conclude this
section by defining reduction of one descriptor to another and show
this is a valid strategy for determining equivalence class order.

\begin{definition}[Reduction]\label{def:reduce}
    Let $A$ and $B$ be two topological descriptor types.
    We say
    $B$ is \emph{reducible} to $A$ if, for all simplicial complexes $K$ and any
    filtration $f$ of $K$, we can compute $\dirDesc{A}{f}$ from
    $\dirDesc{B}{f}$ alone.\footnote{In this reduction, we assume
    the real-RAM model of computation.}
\end{definition}

Intuitively, $B$ is at least as informative as
$A$.
More formally, we have the following lemma:
\begin{restatable}{lemma}{reduce}\label{lem:reduce}
    Let $A$ and $B$ be two topological descriptor types. If $B$ is
    reducible to $A$, we have $\st{A} \preceq \st{B}$.
\end{restatable}

\begin{proof}
    Let $K$ be a simplicial complex.  Define the sets~
    \begin{align*} 
        W_A:= \{P \text{   s.t.   }
        \descSet{A}{K}{P} \text{ is faithful}\} \quad \text{   and   } \quad
        W_B:= \{P \text{   s.t.   }
    \descSet{B}{K}{P} \text{ is faithful}\}.
    \end{align*}
    For each~$P \in W_A$, by definition,~$\descSet{A}{K}{P}$ is faithful.
    Since~$B$ is reducible to~$A$, this also means $\descSet{B}{K}{P}$ is
    faithful, and so~$P$ is also in $W_B$.  Hence,
    $W_A \subseteq W_B$. Hence,~$\inf_{P \in W_A} |P|  \geq \inf_{P \in W_B} |P| $.
    Note that these are exactly the infimums in
    \defref{preceq}, and so, we have~$\st{A} \preceq \st{B}$.
\end{proof}

\section{A Proof of Partial Order}\label{sec:ordering}
In this section, we provide a partial order on the
six topological descriptors of \ssecref{standard}.
While the results and definitions of previous sections were general,
we now focus on descriptors corresponding to lower-star
filtrations.

By simple reduction arguments,
we immediately have the following lemma.

\begin{restatable}{lemma}{twogroups}\label{lem:twogroups}
    $\st{\ecc} \preceq \st{\bc} \preceq \st{\pd}$ and $\st{\aecc} \preceq
    \st{\abc} \preceq \st{\apd}$.
\end{restatable}

\begin{proof}
    The proof follows directly from a reduction argument. We can reduce any
    $\dirDesc{\pd}{\dir}$  to~$\dirDesc{\bc}{\dir}$ by ``forgetting'' the
    relationship between birth and death events. We can then reduce
    $\dirDesc{\bc}{\dir}$ to~$\dirDesc{\ecc}{\dir}$ by taking the alternating
    sum of points from $\dirDesc{\bc}{\dir}$. A nearly identical
    argument shows the relationship between verbose versions of
    these descriptors.
\end{proof}

The reductions described above are well-known, and are observed in other
work; for example~\cite[Prop 4.13]{curry2022many} points out the reduction from
a \pdname to an \eccname.

We also use reduction to order a class of
a concise descriptor type and its verbose counterpart.

\begin{restatable}{lemma}{concise}\label{lem:concise}
    $\st{\ecc} \preceq \st{\aecc}$, $\st{\bc} \preceq \st{\abc}$, and $\st{\pd}
	\preceq \st{\apd}$.
\end{restatable}

\begin{proof}
    Each verbose descriptor has a clear reduction to its concise counterpart. A
	verbose persistence diagram becomes a concise persistence diagram
        by removing all on-diagonal
	points. Verbose Betti functions and verbose Euler characteristic
	functions become concise if we subtract their second coordinates from
	their first coordinates.  Then, by \lemref{reduce}, we have the desired
	relations.
\end{proof}

Next, we see that no concise class is equal to a verbose class.

\begin{restatable}{lemma}{PDneqAECC}\label{lem:PDneqAECC}
    For $\gen \in \{\ecc, \bc, \pd\}$ and $\auggen \in \{\aecc, \abc, \apd\}$,
    we have $\st{\gen} \neq \st{\auggen}$.
\end{restatable}

\begin{proof}
	To show inequality of strength classes, we find a simplicial complex
	for which minimum faithful sets of type $\gen$ and $\auggen$ have
	different cardinalities.  Let $K$ be the simplicial complex that is a
	single edge in $\R^2$ with vertex coordinates $(1,1)$ and $(1,2)$. See
	\figref{verboseneqconcise} for this complex, and an illustration for
        the specific case $\auggen = \apd$. In the direction $e_1=(1,0)$, if $\auggen
	= \apd$, we see an instantaneous birth/death and an infinite birth in
	degree zero. If~$\auggen = \abc$, we see two positive simplices and one
	negative simplex for Betti zero. If~$\auggen = \aecc$, we see two even
	simplices and one odd simplex. This all occurs at height 1,
	and there are no other events, which is only explainable by the
	presence of a single edge.  From $\dirDesc{\auggen}{e_2}$, we see a
	non-instantaneous and instantaneous event at heights 1 and 2,
	respectively, which give us the~$y$-coordinates of our two vertices.
	Then, $\simComp$ is the only complex that could have generated both
	$\dirDesc{\auggen}{e_1}$ and~$\dirDesc{\auggen}{e_2}$, i.e., the set
	$\descSet{\auggen}{K}{\{e_1, e_2\}}$ is faithful.

    Next, consider the descriptor type $\gen$.  For any $\dir \in \sph^1$,
	$\dirDesc{\gen}{s}$ contains exactly one event; if the
	lowest vertex of $K$ with respect to $s$ has height $a$ in
	direction $s$, then $\dirDesc{\gen}{s}$ records a change in
	homology/Betti number/Euler characteristic at height $a$ and records no
	other changes.
	Thus, $\dirDesc{\gen}{s}$ can only give us information about one
	coordinate of the vertex set of $K$ at a time, corresponding to
	whichever vertex is lowest in direction $s$. However,~$\simComp$ has
	three relevant coordinates; namely, $x=1$,~$y=1$, and~$y=2$, meaning it
	is not possible for any faithful set of type $\gen$ to have size less
	than three. Thus, since $3 \neq 2$, we have shown $\st{\gen} \neq
	\st{\auggen}$, as desired.
\end{proof}

\begin{figure}[h!]
    \centering
	\includegraphics[scale=.45]{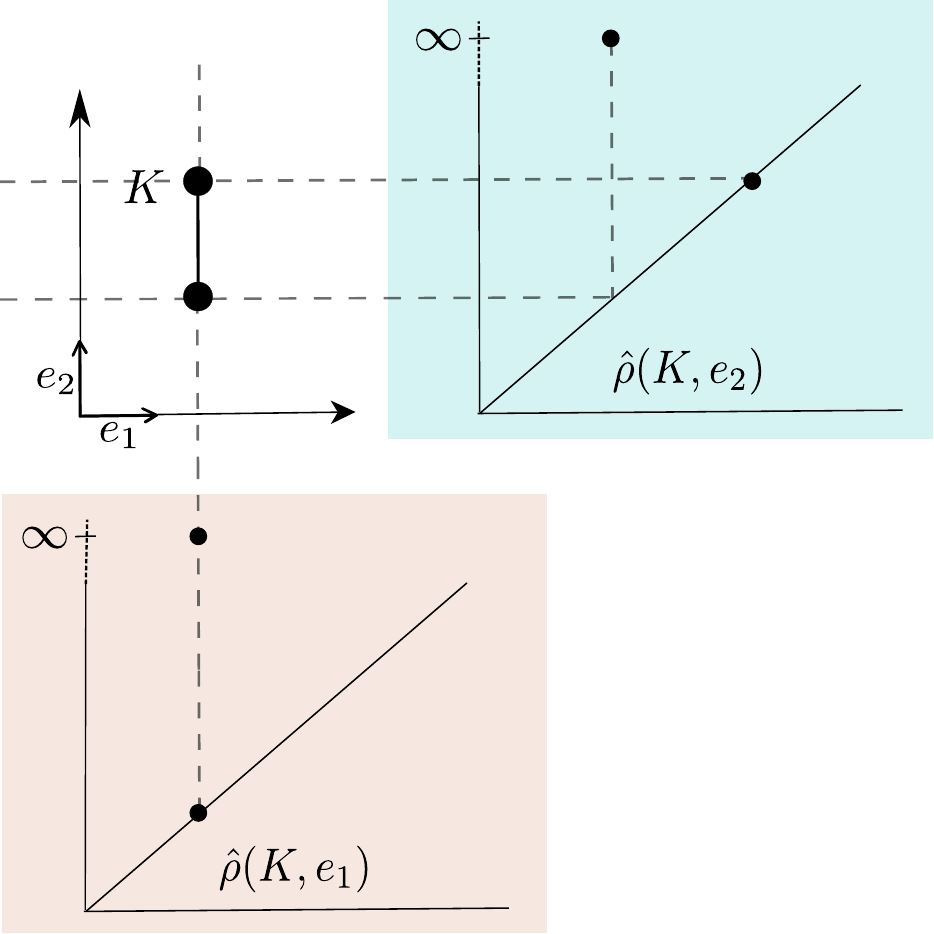}
	\caption{The simplicial complex considered in the proof of
	\lemref{PDneqAECC} as well as the verbose diagrams in directions~$e_1$
	and $e_2$. Note that $\apd(K,\{e_1, e_2\})$ is indeed a faithful set,
	since we can recover the coordinates of both vertices ($\apd(K, e_1)$
	tells us the $x$-coordinates and $\apd(K, e_2)$ tells us the
	$y$-coordinates), as well as determine there is only a single edge
	present (there is only one instantaneous zero-dimensional point in each
	verbose diagram). The concise versions of these diagrams do not have
	on-diagonal points, and each only contain a single point
	at~$\infty$. This is true for concise diagrams corresponding to
        any direction.
	}
\label{fig:verboseneqconcise}
\end{figure}

The specific inequality $\st{\ecc} \neq \st{\apd}$ is also implied by
\cite[Thm. 10]{mickaPhD}.

\begin{figure}[h!]
    \centering
    \begin{subfigure}[b]{0.3\textwidth}
        \centering
        \includegraphics[scale=.6]{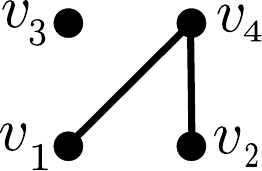}
	    \subcaption{$\simComp$}
	    \label{subfig:abcvapdK}
    \end{subfigure}
    \begin{subfigure}[b]{0.3\textwidth}
        \centering
        \includegraphics[scale=.6]{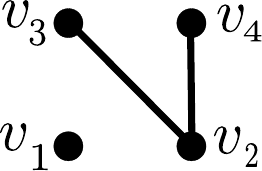}
	    \subcaption{$\simComp'$}
	    \label{subfig:abcvapdKprime}
    \end{subfigure}
	\caption{Complexes used in the proof
        of \lemref{ABCleqAPD}.}
\label{fig:abcvapd}
\end{figure}

In~\cite{fasy2022efficient}, faithfulness is shown via knowing the dimension of
event rather than any birth-death pairings, so the resulting faithful sets of
verbose persistence diagrams and verbose Betti functions
have equal cardinality for every simplicial complex. One
might wonder, then, if $\st{\apd}$ equals~$\st{\abc}$.  However, this is an
incorrect leap; faithful sets of~\cite{fasy2022efficient} are almost certainly
not minimal.  The next lemma gives an instance where birth-death pairings
matter, rather than just an event's existence or dimension.

\begin{lemma}\label{lem:ABCleqAPD}
    $\st{\abc} \prec \st{\apd}$.
\end{lemma}

\begin{proof}
    We know by \lemref{twogroups} that $\st{\abc} \leq \st{\apd}$. We must
    show that equality does not hold; that is, that there exists a simplicial complex for
    which the cardinality of the minimal faithful sets of \abcnames and \apdnames
    differ.
    Consider the simplicial complex $K$ in $\R^2$ consisting of: four
    vertices~$v_1 = (0,0),
    v_2 = (0,1), v_3 = (1,0)$, and $v_4 = (1,1)$ and two edges $[v_1, v_4]$
    and $[v_2, v_4]$, as in \subfigref{abcvapd}{abcvapdK}.

    Let $e_1=(1,0)$ and $e_2=(0,1)$.
        We first claim that $\dirDescLong{\apd}{K}{\{e_1, e_2\}}$ is faithful,
        meaning $\Gamma(K, \apd) \leq 2$.  Such diagrams uniquely identify the
        vertex set of $K$ by~\cite[Lemma 4]{belton2019reconstructing}; we provide further
        details here. From $\apd(e_1)$, we know
	$K$ has four vertices, two
	with~$x$-coordinate 0, and two with~$x$-coordinate 1. Similarly,
        from $\apd(e_2)$, we know two of the four vertices
        have $y$-coordinate 0 and two have $y$-coordinate 1.
	There are exactly four ways to pair our~$x$-
	and~$y$-coordinates, so we know the locations of each
        vertex.  See \subfigref{pairings}{correctpairings}.

    For the edges, we first note that
    from~$\dirDesc{\apd}{e_2}$ in degree zero we see an
    instantaneous birth/death
    at height one as well as a connected component born at
    height zero that dies at height one, so we know $K$ has exactly
    two edges with  height one in direction $e_2$.
    Namely, we know we have either the edges~$[v_1, v_3]$ and~$[v_2, v_3]$,
    or~$[v_1, v_4]$ and~$[v_2, v_4]$, i.e., we have one of the two
    complexes shown in \figref{abcvapd}.
    Because $\dirDesc{\apd}{e_1}$ sees
    two
    zero-dimensional births at height zero with an infinite lifespan, we
    know there is no edge from~$v_1$ to~$v_3$.
    Finally, since higher homology is trivial, we know there are no other
    simplices and have determined $K$ exactly; thus, we have a
    faithful set of size two.

    We next show that $\Gamma(K,\abc)>2$.
    Suppose, by way of contradiction, that $s_1$ and~$s_2$ are
    two directions such that~$\descSet{\abc}{K}{\{s_1, s_2\}}$ is a
    faithful set.  We first show, without loss of generality,~$s_1 \in
    \{e_1, -e_1\}$ and~$s_2 \in \{e_2, -e_2\}$. Suppose this is not the
    case.  Because~$s_1 = -s_2$ does not correspond to a faithful
    set, we assume (without loss of generality) that~$s_1 \neq -s_2$.  Then, at least one of $s_1$
    or~$s_2$ sees the vertices of $K$ at more than two distinct heights;
    see \subfigref{pairings}{toomanypairings}.
    In
    order to know the precise coordinates of each vertex, we need to
    correctly pair heights in directions $s_1$ and~$s_2$. However,
    since at least one of $s_1$ or $s_2$ reports more than two distinct
    heights, we have more than four possible pairings (see
    also~\cite[Lemma~4]{belton2019reconstructing}). We claim it is not possible to find the
    four correct pairings.  The degree-zero information~$\abc_0(s_1)$
    and~$\abc_0(s_2)$ alone is insufficient, as it only tells us the
    heights of vertices. From $\abc_1$, we know the height of edges, which
    only confirms the height of the top vertex and that there is some
    vertex below, information we already had from~$\abc_0$.
    Thus, we must
    have~$s_1 \in \{e_1, -e_1\}$ and $s_2 \in
    \{e_2, -e_2\}$.  However, for each of these four directions, the
    associated verbose Betti function is not able to distinguish the two
    complexes shown in
    \subfigref{abcvapd}{abcvapdKprime}. For instance,
    both~$\dirDescLong{\abc}{K}{e_1}$ and $\dirDescLong{\abc}{K'}{e_1}$ see two
    vertices at height zero, and two vertices and two edges at height one,
    i.e.,~$\abc(K, e_1) = \abc(K', e_1)$. The other cases
    of~$s_1$ and~$s_2$ are similar.
    Thus, we have found a faithful set of verbose persistence diagrams
    with cardinality two, but have
    shown any faithful set of verbose Betti functions must have cardinality
    greater than
    two.
\end{proof}

\begin{figure}[h!]
    \centering
    \begin{subfigure}[b]{0.3\textwidth}
        \centering
        \includegraphics[scale=.6]{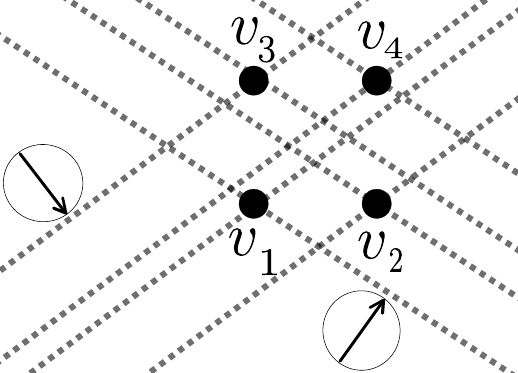}
        \subcaption{}
	    \label{subfig:toomanypairings}
    \end{subfigure}
    \begin{subfigure}[b]{0.3\textwidth}
        \centering
        \includegraphics[scale=.6]{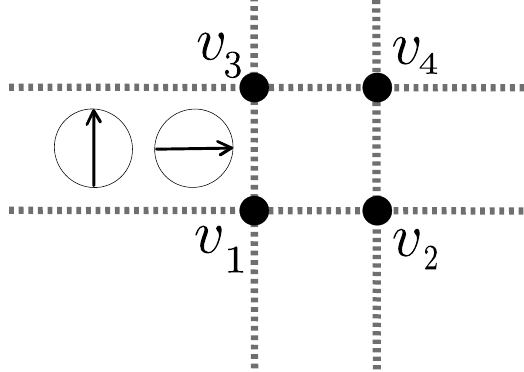}
            \subcaption{}
	    \label{subfig:correctpairings}
    \end{subfigure}
	\caption{For the given vertex set, heights of
        filtration events in the indicated directions are shown as
	dashed grey lines. While we know the number of vertices on each
        line, for two directions not both in $\{\pm e_1, \pm e_2\}$, as
        in~\figpartref{toomanypairings},
        we can not identify vertex locations. Only
        when choosing one each of $\pm e_1$ and $\pm e_2$, as
        in~\figpartref{correctpairings}, is the set of
	vertices satisfying these constraints unique.
        }
\label{fig:pairings}
\end{figure}

Combining results, we arrive at our main theorem.

\begin{theorem}[Partial Ordering]
    \label{thm:order}
    The partial order of strength classes of topological descriptor types shown
    in \figref{abstract} is correct.
\end{theorem}

\section{Bounds on Faithful Sets}\label{sec:bounds}
Here, we provide lower bounds on the size of faithful sets of the six
descriptor types of \ssecref{standard}.

\subsection{Concise Descriptor Bounds}\label{ss:concise}
A defining feature of concise descriptors is that there are not generally
events at every vertex height in a filtration.
The closer a feature is to coplanar,
the smaller the range of directions that can detect it becomes
(\hspace{1sp}\cite[Sec. 4]{fasy2018challenges} explores this specifically for Euler
characteristic functions). Difficulty detecting the presence or absence of
structures near to the
same affine subspace puts greater restrictions on the ability of concise
descriptors to form faithful sets.  We use the following definition
to help this claim be precise.

\begin{definition}[Simplex Envelope]\label{def:envelopes}
    Let $\simComp$ be a simplicial complex in $\R^d$, let $\sigma \in
	\simComp$, and let~$S \subseteq \S^{d-1}$.
    Then, we define the \emph{envelope of $\sigma$}, denoted $\E^S_\sigma$,
as the intersection
    of (closed) supporting halfspaces
\[
\E_\sigma^S = \bigcap_{s \in S} \{p \in \R^d \mid s \cdot p \geq \min_{v \in
    \sigma}(s \cdot v)\}.
\]
    If $S$ is clear from context, we write $\E_{\sigma}$. By the \emph{dimension}
    of $\E_\sigma$, we mean the largest dimension of ball that can be contained
    entirely in $\E_\sigma$.
\end{definition}

See \figref{envelope}.
\begin{figure}[h!]
    \centering
    \includegraphics[width=.35\textwidth]{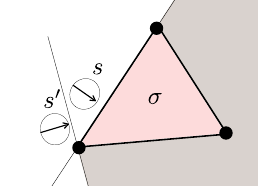}
    \caption{
        Given $S = \{s, s'\}$, the envelope for $\sigma$ is the grey and pink
        shaded regions. 
    }
\label{fig:envelope}
\end{figure}
Since $\E_\sigma^S$ is an intersection of convex regions, it is itself
convex. Furthermore, with respect to each $s \in S$, the height of
each point of $\sigma$ is greater than or equal to its minimum vertex, so
$\E_{\sigma}^S$ contains~$\sigma$.
\begin{remark}
        The simplex envelopes of \defref{envelopes} have connections to
        well-studied topics such as convex cones, support functions, etc.
        See~\cite{darrotto2013convex, panik2013fundamentals}. In
        particular,~\cite[Thm 3.1.1, Cor 3.1.2]{panik2013fundamentals}
        establish that a simplex envelope
        corresponding to the entire sphere of directions is the simplex itself.
\end{remark}
We use simplex envelopes to define a necessary condition for
concise descriptors to form a faithful set.

\begin{restatable}[Envelopes for Faithful Concise Sets]{lemma}{envelopes}\label{lem:envelopes}
    Let $\simComp$ be a simplicial complex immersed in~$\R^d$, let~$D \in \{\ecc,
    \bc, \pd\}$, and let $S \subseteq \sph^{d-1}$ so that
    $\descSet{\gen}{K}{S}$ is faithful. Then, for any maximal simplex~$\sigma$
    in $\simComp$, the dimension of $\E_\sigma$ equals the dimension of
    $\sigma$.
\end{restatable}

\begin{proof}
    Let $k$ be the dimension of $\sigma$, and let $c$ be the dimension
	of $\E_\sigma$.  First, we observe that since $\sigma$ is contained in
	$\E_\sigma$, we must have $k \leq c$. The claim is trivial when $k = d$,
	so we proceed with the case~$k < d$ and assume, by way of contradiction,
	that $k < c$.

    We claim that in this case, 1) at every interior point of $\sigma$, there
    is a vector normal to $\sigma$ that ends in the interior of $\E_\sigma$, and
    2) letting $p$ denote the endpoint of such a vector, $p$ is higher than the
    lowest vertex of $\sigma$ with respect to each $s \in S$.

    The first part of the claim, 1), is true since otherwise, $\E_\sigma$ would
    be $k$-dimensional. 2) is true since each halfspace defining $\E_\sigma$
    contains the lowest vertex (or vertices) of $\sigma$ with respect to the
    corresponding direction. Thus, since $p$ is in the interior of $\E_\sigma$,
    it must be higher than this lowest vertex (or vertices) with respect to the
    corresponding direction.

    Now consider the simplicial complex $K'$, defined as having all the
    simplices of $K$ in addition to the simplex formed by taking the geometric join of
	$p$ and $\sigma$,
    i.e., the simplex~$p \ast \sigma$.
    We claim that, for any~$s \in S$, we have~$\descSet{\gen}{K'}{s} =
    \descSet{\gen}{K}{s}$.
    First, we note that since~$p \ast \sigma$ deformation retracts
	onto~$\sigma$, $K'$ has the same homology as $K$. Next, we observe
        that~$\descSet{\gen}{K'}{s}$ and~$\descSet{\gen}{K}{s}$ cannot differ by
	more than a connected component birth/death; higher dimensional
	differences would require more than the join of a point with an
	existing face.

    Finally, since $p$ is higher than the lowest vertex of $\sigma$ with respect
    to any direction $s \in S$, the simplex~$p \ast \sigma \in K'$ does not correspond to
    any connected component birth or death in~$\descSet{\gen}{K'}{s}$ that was
    not present in~$\descSet{\gen}{K}{s}$.  Thus, we have shown
    $\descSet{\gen}{K'}{S} = \descSet{\gen}{K}{S}$.  This contradicts the
    assumption that~$\descSet{\gen}{K}{S}$ is faithful, so we must have~$k = c$.
\end{proof}

\begin{figure}[h!]
    \centering
    \begin{subfigure}[b]{0.3\textwidth}
        \centering
        \includegraphics[width=.9\textwidth]{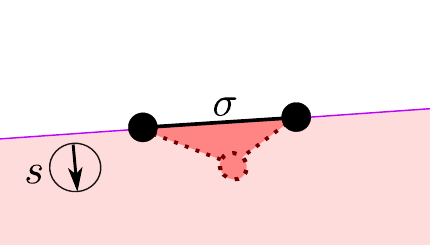}
	    \subcaption{}
        \label{subfig:edge_envelope1}
    \end{subfigure}
    \begin{subfigure}[b]{0.3\textwidth}
        \centering
        \includegraphics[width=.9\textwidth]{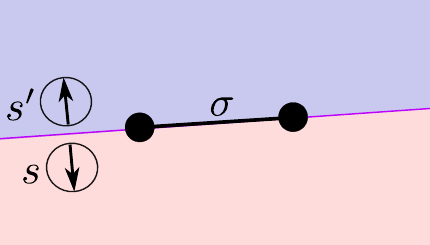}
	    \subcaption{}
        \label{subfig:edge_envelope2}
    \end{subfigure}
	\caption{With only the single direction~$s$
	perpendicular to maximal edge $\sigma$ in $\R^2$,
	the envelope $\E_\sigma^s$ is two-dimensional.
	Then, we could
	place an adversarial two-simplex contained in $\E_\sigma^s$
	that is undetectable by~$D(s)$, for~$D \in \{\ecc,
	\bc, \pd\}$, as in
        \figpartref{edge_envelope1}.
   In \figpartref{edge_envelope2}, the inclusion of $s'$ reduces
    $\E_\sigma^{\{s, s'\}}$ to a linear subspace (purple intersection of pink
    and blue halfspaces) and the adversarial two-simplex would be detected by
    $D(s')$.
    }
\label{fig:3Delbow}
\end{figure}

See \figref{3Delbow} for an example of what might go wrong if a simplex
envelope does not satisfy the conditions of \lemref{envelopes}.
Since we require the envelopes of a $k$-simplex to be
$k$-dimensional, and since envelopes are the intersections of closed half
spaces, standard arguments from manifold theory give us the following.

\begin{corollary}[Concise Descriptors Per Maximal Simplex]\label{cor:perp}
    Let $\simComp$ be a simplicial complex immersed in~$\R^d$, let~$\gen \in \{\ecc, \bc,
    \pd\}$, and let $S \subseteq \sph^{d-1}$.  If $\descSet{\gen}{K}{S}$ is
    faithful, then for each maximal $k$-simplex~$\sigma \in \simComp$ with
    $k<d$, the set $S$ has at least $d-k+1$ directions
    perpendicular to $\sigma$.
    If $k<d-1$, these
    directions are pairwise linearly independent.
\end{corollary}

\lemref{envelopes} and \corref{perp} each give us the following.

\begin{corollary}[Tight Lower Bound]\label{cor:strictconcise}
Let $\simComp$ be a simplicial complex in $\R^d$, $\gen \in \{\ecc, \bc,
	\pd\}$, and $S \subseteq \sph^{d-1}$. Suppose that
	$\descSet{\gen}{K}{S}$ is faithful. Then, $|S| \geq d+1$, and this bound
	is tight.
\end{corollary}

This bound is met whenever $\simComp$ is a single vertex.
However, minimal faithful sets of concise
descriptors are generally much larger.
Counteracting the need for perpendicular
directions is the fact that, as~$d$ increases, more simplices span common
hyperplanes, so perpendicular directions can increasingly be shared.
We use these observations to lower bound the worst-case size of faithful
set of concise descriptors.

\begin{theorem}[Lower Bound for Worst-Case Concise Descriptor Complexity]
    \label{thm:highlowbound}
    Let $\gen \in
    \{\ecc, \bc, \pd\}$ and let~$K$ be a simplicial complex in $\R^d$ with
	$n_1$ edges. Then, the worst-case cardinality of a minimal 
        faithful descriptor
	set of type $\gen$ is $\Omega(d + n_1)$.
\end{theorem}

\begin{proof}
	We construct a simplicial complex, $K$, and bound the
        minimum cardinality of a faithful set for~$K$.  Suppose that, for
        $d>2$, that $\simComp$ is a graph in $\R^d$ with~$n_1 < d-1$ edges, and
        for some $S \subseteq \sph^{d-1}$, the set~$\descSet{\gen}{K}{S}$ is
        faithful.  Then, by \lemref{envelopes}, the envelope of each maximal
        edge $\sigma$ must be one-dimensional. Then, by \corref{perp}, for every
        such $\sigma$,~$S$ contains~$d-1+1 = d$ pairwise linearly independent
        directions perpendicular to~$\sigma$.  Let $S^\ast$ be a minimal subset
        of directions in $S$ satisfying the conditions of perpendicularity and
        one-dimensional envelopes.

	To build $S^\ast$, first note all edges of $K$
	are contained in a common~$n_1$-plane, so there is a $(d-n_1 -
	1)$-sphere's worth of directions perpendicular to \emph{all} edges
	simultaneously. Such directions are maximally
	efficient in the sense that each can ``count'' for
	all edges at once. We choose any $d-1$ pairwise linearly independent
        directions from this sphere to be included in $S^\ast$.
        Now we need
	an additional perpendicular direction for each edge to bring the total
	for each edge to $d$.  To ensure
	the envelopes of each edge are one-dimensional, these additional
	directions must not be perpendicular to any hyperplane defined by
	subsets of more than one edge. This means we must consider a total of
	$n_1$ additional directions, so that~$S^\ast$ has cardinality~$d-1 +
	n_1$.  Since~$|S^\ast|$ lower bounds $|S|$, we find $|S| \in
	\Omega(d+n_1)$.
\end{proof}

%

\subsection{Verbose Descriptor Bounds}\label{ss:verbose}
We now shift to verbose descriptors, and
begin with the tight lower bound.
\begin{lemma}[Tight Lower Bound]\label{lem:strictaug}
Let $\simComp$ be a simplicial complex in $\R^d$ and $\auggen \in \{\aecc,
    \abc, \apd\}$. Suppose for some $S \subseteq
    \sph^{d-1}$ the set~$\descSet{\auggen}{K}{S}$ is faithful. Then,~$|S| \geq d$, and
    this bound is tight.
\end{lemma}

\begin{proof}
No vertex in $\simComp$ can be described using fewer than $d$ coordinates, so
    a set of descriptors of type~$\auggen$ with cardinality less than $d$ can never
be faithful. To see that this bound is tight, when
$\simComp$ is a single vertex, verbose descriptors generated by any $d$
pairwise linearly independent directions form a faithful set.
\end{proof}

Next, we identify a family of simplicial complexes for which minimal faithful
sets of verbose descriptors are at least linear in the
number of vertices.  We use~$\alpha_{i,j}$ to denote the angle that vector~$v_j
- v_i$ makes with the $x$-axis, taking value in~$[0, 2\pi)$.  We first observe
a consequence of a specific instance of the general phenomenon that a
simplicial complex stratifies the sphere of directions based on vertex
order~\cite{curry2022many,leygonie2019framework}.

\begin{observation}\label{obs:semicircle}
    Suppose a simplicial complex in $\R^2$ contains an isolated
    edge~$[v_1, v_2]$.
        Then, a
	birth event occurs height $s \cdot v_1$ in
	$\dirDescLong{\apd}{\simComp}{\dir}$ for all~$\dir$
	the interval~$I = (\alpha_{1,2} - \pi/2, \alpha_{1,2}
	+ \pi/2)$ of $\sph^1$ (all $s$ so that~$s \cdot v_1 > s \cdot v_2$) and as an
	instantaneous event for all~$\dir \in I^C$.
\end{observation}

Next, we make the following geometric observation.

\begin{observation}\label{obs:extremal}
    Consider a pair of nested triangles as in \figref{angles}.
    Then, angle $A$ is larger than $\theta$,~$\phi - B$, and~$\psi - C$.
\end{observation}
\begin{figure}[h!]
    \centering
    \includegraphics[width=.25\textwidth]{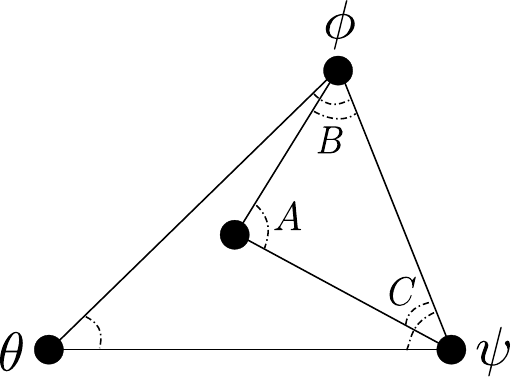}
    \caption{Nested triangles as discussed in \obsref{extremal}}
\label{fig:angles}
\end{figure}

We now construct the building block that forms the complexes used in our bound.

\begin{construction}[Clothespin Motif]\label{cons:clothespin}
Let $\simComp$ be a simplicial complex in $\R^2$ with a vertex set~$\{v_1, v_2,
v_3, v_4\}$ such that only $v_3$ is in the interior of the convex hull of
    $\{v_1,v_2,v_4\}$, and that the edge set consists of~$[v_1, v_2]$ and $[v_3,
    v_4]$. See \figref{clothespinK}.
\end{construction}

\begin{figure}[h!]
    \centering
    \begin{subfigure}[b]{0.3\textwidth}
        \centering
        \includegraphics[scale=.6]{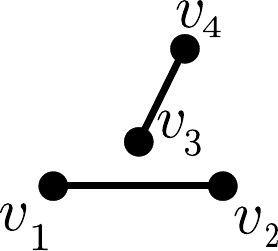}
	    \subcaption{$\simComp$}
	    \label{fig:clothespinK}
    \end{subfigure}
    \begin{subfigure}[b]{0.3\textwidth}
        \centering
        \includegraphics[scale=.6]{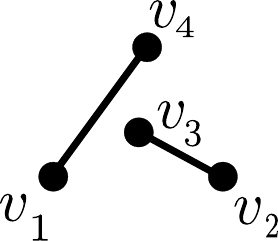}
	    \subcaption{$\simComp'$}
	    \label{fig:clothespinKprime}
    \end{subfigure}
    \caption{The two simplicial complexes of
    \lemref{clothespin}.}
\label{fig:kandkprime}
\end{figure}

\consref{clothespin} was built specifically for the following necessary
condition:

\begin{restatable}[Clothespin Representability]{lemma}{clothespin}\label{lem:clothespin}
    Let~$\simComp$ be as in \consref{clothespin}, and suppose
    that $\descSet{\apd}{K}{S}$ is faithful. Then, there is
    some~$\dir \in S$ so that the angle formed between $\dir$ and
    the $x$-axis
    lies in the region
    $$W = [\alpha_{3,2} - \pi/2, \alpha_{3,4} - \pi/2] \cup
    [\alpha_{3,2} + \pi/2, \alpha_{3,4} + \pi/2].
    $$
\end{restatable}

\begin{proof}
    Let $\simComp'$ be a simplicial complex immersed in $\R^2$ with the same vertex set as
    $\simComp$, but with edges~$[v_1, v_4]$ and~$[v_2, v_3]$ (see
    \figref{clothespinKprime}).
    Recall that, since $\descSet{\apd}{K}{S}$ is faithful, by definition, the set
    $S$ must contain some direction~$s$ so that~$\descSet{\apd}{K}{s} \neq
    \descSet{\apd}{K'}{s}$.

    Each vertex corresponds to either a birth event or an instantaneous event
    depending on the direction of filtration.
    We proceed by considering each vertex $v_i$ individually and determining subsets
    $R_i \subset \sph^1$ such that, whenever $s \in R_i$, the event at $\dir
    \cdot v_i$ is different when filtering over $\simComp$ versus $\simComp'$,
    but for~$s_* \not \in R_i$, the type of event at $\dir_* \cdot v_i$ is the
    same between the two graphs.  \figref{observability} shows these regions,
    and in what follows, we define them precisely.

    First, consider $v_1$.
    By~\obsref{semicircle}, $v_1 \in \simComp$ corresponds to a birth event for all
    directions in the interval $B = (\alpha_{1,2}- \pi/2, \alpha_{1,2}+ \pi/2)$ and
    $v_1 \in \simComp'$ corresponds to a birth event for all directions in the
    interval $B' = (\alpha_{1,4} - \pi/2, \alpha_{1,4} + \pi/2)$. Then, we write $R_1 =
    (B \setminus B') \cup (B' \setminus B)$, which is the wedge-shaped region such that
    for any $\dir \in R_1$, the type of event associated to~$v_1 \in \simComp$ and
    $v_1 \in \simComp'$ differ, meaning~$\dirDescLong{\apd}{\simComp}{\dir} \neq \dirDescLong{\apd}{\simComp'}{\dir}$.

    Using this same notation, identify the wedge-shaped region $R_i$ for vertex
    $i \in [2,3,4]$ such that any
    direction from $R_i$ generates verbose persistence diagrams
    that have different event types at the height of vertex
    $v_i$ when filtering over $\simComp$ versus~$\simComp'$. Similar
    arguments for $i \in [2,3,4]$ give us the complete list;
    \begin{align*}
        R_1 &= (\alpha_{1,2} - \pi/2, \alpha_{1,4} - \pi/2] \cup [\alpha_{1,2}+
        \pi/2,
        \alpha_{1,4} + \pi/2) \\
        R_2 &= (\alpha_{2,3} - \pi/2, \alpha_{2,1} - \pi/2] \cup [\alpha_{2,3} +
        \pi/2,
        \alpha_{2,1} + \pi/2) \\
        R_3 &= (\alpha_{3,2} - \pi/2, \alpha_{3,4} - \pi/2] \cup  [\alpha_{3,2}
        + \pi/2,
        \alpha_{3,4} + \pi/2) \\
        R_4 &= (\alpha_{1,4} - \pi/2, \alpha_{3,4} - \pi/2] \cup [\alpha_{1,4} +
        \pi/2, \alpha_{3,4} + \pi/2)
    \end{align*}
    Let $W = \cup_{i = 1}^4 R_i$, meaning $W$ is the set of directions for which
    the corresponding filtrations have different event types at some vertex of
    $K$ and $K'$.
    Then, for any $\dir \in W$, we have
    $\dirDescLong{\apd}{\simComp}{\dir} \neq
    \dirDescLong{\apd}{\simComp'}{\dir}$, and for any $\dir_* \in W^C$, we
    have~$\dirDescLong{\apd}{\simComp}{\dir_*} = \dirDescLong{\apd}{\simComp'}{\dir_*}$.

    Finally, we claim that $W$ is the closure of ${R_3}$, denoted
    $\overline{R}_3$, i.e., exactly the region
    described in the lemma statement.  This is a direct corollary to
    \obsref{extremal}; the angles swept out by each region correspond to
    the angles formed by pairs of edges in $K$ and $K'$; in particular, the
    angle $\measuredangle v_2 v_3 v_4$ is the largest and geometrically
    contains the others.  This means the extremal boundaries over all
    $R_i$'s are formed by the angles $\alpha_{2,3} \pm \pi/2$ and
    $\alpha_{3,4} \pm \pi/2$, the defining angles of $R_3$. Each of these
    four angles appears as an included endpoint for some
    $R_i$, so~$R_1, R_2, R_4 \subseteq \overline{R}_3 = W$
    (\figref{observability}) and we have shown our claim.
\end{proof}

\begin{figure}[h!]
    \centering
    \includegraphics[width=.35\textwidth]{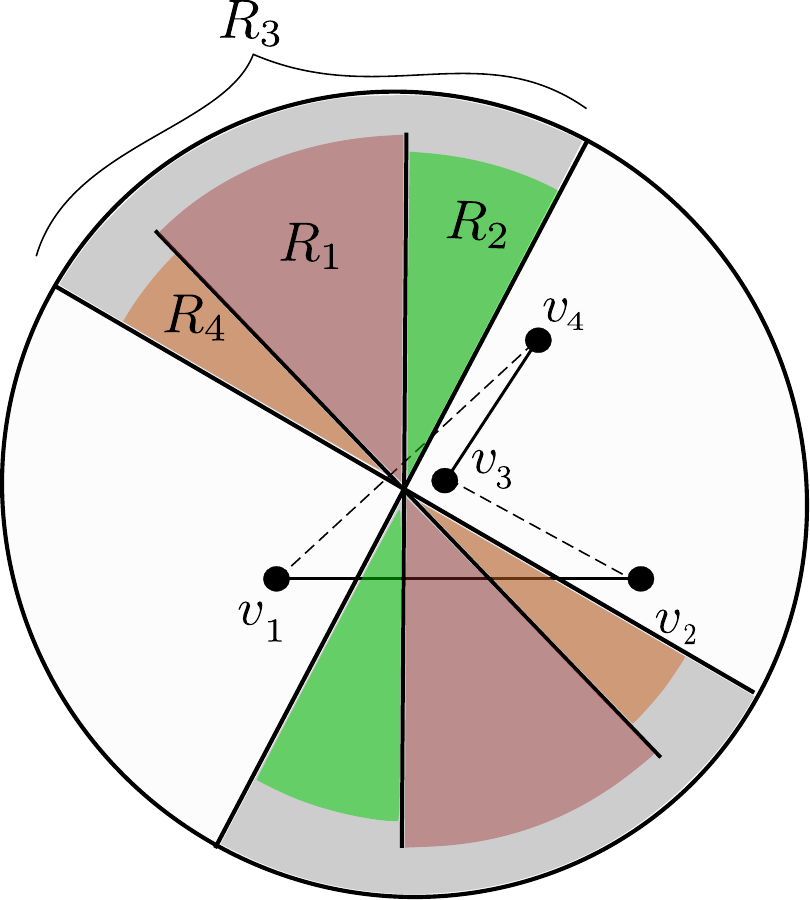}
    \caption{The regions described in the proof of \lemref{clothespin}, with
	additional shading in the interior of the sphere of directions to aid
	in visibility. $K$ is shown as
    solid black edges and~$K'$ as dashed edges. For any lower-star filtration in
    a direction contained in $R_i$, the
    event at vertex~$v_i$ differs when considering~$K$ or~$K'$, thus, such
    directions are able to distinguish $K$ from $K'$. Note that any direction
    outside the regions of observability (i.e., the non-shaded portions of the
    circle) is not able to distinguish $K$ from $K'$.}
\label{fig:observability}
\end{figure}

We call $W$, the intervals of directions in $\sph^1$ for which
corresponding verbose descriptors can distinguish $K$ from $K'$ a
clothespin's \emph{region of observability} (similar to observability for
concise Euler characteristic functions
in~\cite{fasy2018challenges,curry2022many}). Crucially,~$W$ is defined
by $\measuredangle v_2v_3v_4$, so we have the following.

\begin{remark}[$W$ Can be Arbitrarily Small]\label{rem:shrinkwedge}
As the angle $\measuredangle v_2 v_3 v_4$ approaches zero, the region of
observability from \lemref{clothespin} also approaches zero.
\end{remark}

We use \remref{shrinkwedge} to piece together clothespins
so their regions of observability do not overlap.

\begin{construction}[Clothespins on a Clothesline]\label{cons:clothesline}
    Let $\simComp^{(m)}$ be a simplicial complex in $\R^2$ formed by $m$ copies of
\consref{clothespin} ($m$ clothespin motifs) such that the regions of
	observability for each clothespin do not intersect. This is
	possible for any $m$ by \remref{shrinkwedge}.
\end{construction}
See \figref{clothesline}. \consref{clothesline} implies a lower bound on the
worst-case cardinality of faithful sets of verbose persistence diagrams, which
we formalize in the following lemma.
\begin{figure}[h!]
    \centering
    \includegraphics[width=.6\textwidth]{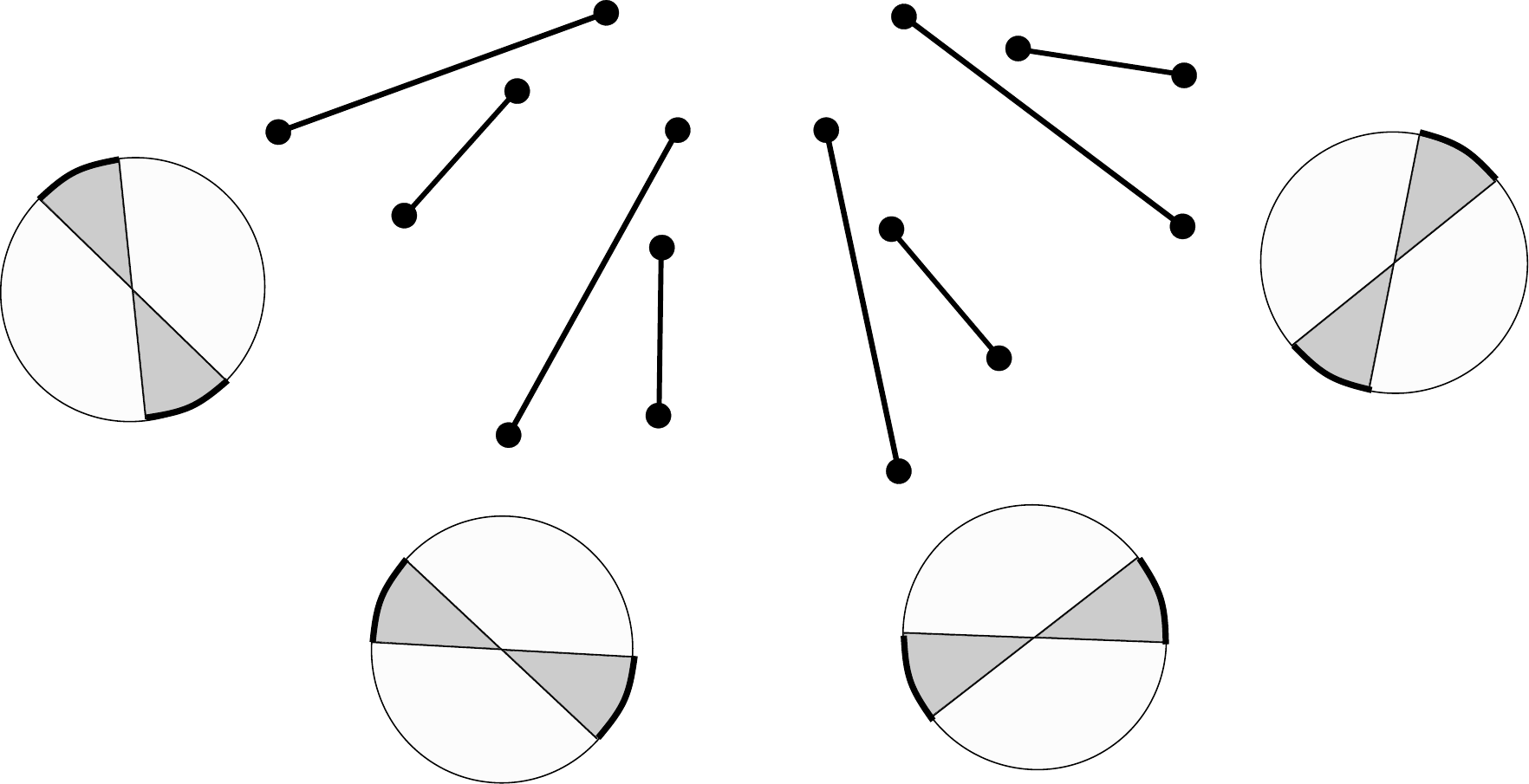}
    \caption{An example of $K^{(m)}$ for $m=4$. Regions of
observability are shown below each clothespin. By construction, each of these
regions of $\sph^1$ are disjoint.}
\label{fig:clothesline}
\end{figure}
\begin{lemma}[Verbose Persistence Diagram Complexity]\label{lem:apdbound}
    Let $\simComp^{(m)}$ be as in \consref{clothesline} and
    suppose~$\descSet{\apd}{K^{(m)}}{S}$ is a faithful set. Then, $S$
    contains at
    least one direction in each of the $m$ regions of
    observability, so $|S| \geq m  = n_0 / 4$.  Thus, $|S|$ is
    $\Omega(n_0)$.
\end{lemma}

By \thmref{order}, \lemref{apdbound} implies the following:

\begin{theorem}[Lower Bound for Worst-Case Verbose Descriptor Complexity]
    \label{thm:allbound}
    Let $\auggen \in \{\aecc, \abc, \apd\}$. Then, the worst-case cardinality of
	a minimal descriptor set of type $\auggen$ is $\Omega(n_0)$.
\end{theorem}

\section{Discussion}\label{sec:discussion}
We provide a framework for comparing general topological descriptor types by
their ability to efficiently represent simplicial complexes. The tools
developed here are a first step towards more theoretical justifications 
for the use of particular descriptor types in applications.

We focus on the descriptors that are particularly relevant to
applications and related work; verbose and concise Euler characteristic
functions, verbose and concise Betti functions, and verbose and concise
persistence diagrams. We give a partial order on this set of six
descriptors, including the strict inequality, $\st{\abc} \prec \st{\apd}$.

We then identify tight lower bounds for both concise and verbose descriptors in
the set of six, as well as asymptotic lower bounds for worst-case complexity of
sizes of faithful sets.  Because faithful sets of concise descriptors require
many perpendicular directions to each maximal simplex, a huge hindrance in
practice, we believe applications research may benefit from the use of verbose
descriptors rather than the current standard of concise descriptors.

Perhaps the strength classes $\st{\ecc}, \st{\bc},$ and $\st{\pd}$ intuitively
feel as though they should be related by strict inequalities.  However, this
issue is nuanced. \lemref{trivial} (\appendref{other-descriptors}) shows the
impact that general position assumptions have on relations in this
set.  But even with general position, the seemingly advantageous ``extra''
information of homology compared to, e.g., Euler characteristic may no longer
be so useful when we require tight envelopes around each maximal
simplex. That is, once we have all the (many) required directions, we have
already carved out the space filled by the complex, and already know quite a
lot simply from the presence of events.  Non-equality/equality of concise
descriptors remains an area active of research.

In other ongoing work, we hope to classify the simplicial complexes for which
minimal faithful sets of verbose descriptors are independent of the size of the
complex. We are also interested in relating other common descriptor types, such
as merge trees.

%
%
\small
\bibliographystyle{abbrv}
\bibliography{references}

\normalsize
\appendix
\section{Example Filtration with Six Descriptor Types}
\label{append:six}
We consider the simplicial complex on four vertices given in \subfigref{six}{six-K}.
In the $e_1$ direction 
we see four distinct heights of vertices,
$a$, $b$, $c$, and $d$. 

First, we describe what happens in $\pd(K)$ and $\apd(K)$. At height
$a$, and then again at height $b$, we see connected components born.  At
height~$c$, the homology of the sublevel set does not change, so no
change is recorded in $\pd(K)$. However, a corresponding index filtration sees
the connected component corresponding to first adding the vertex at $c$, which
then immediately dies as we include the edge at height $c$. Thus, in
$\apd(K)$, we have the point $(c,c)$. For similar reasons, we see the point
$(d,d)$ in $\apd(K)$. Also at height~$d$, our two connected components
merge into a single connected component. It is a standard convention to 
choose the eldest component to survive, so we have the point $(b,d)$ in both
diagrams. We also see a cycle appear at height $d$, giving us the point
$(d, \infty)$ in both diagrams.  Finally, since the connected component born at
height~$a$ did not die, we have the point $(a, \infty)$ in both diagrams.

Next, we describe what happens in $\bc(K)$ and $\abc(K)$.  At the height $a$,
only the Betti number in dimension zero changes, going from zero to one. Since
the inclusion of this vertex increased Betti zero, we count the vertex as
positive in $\abc_0(K)$. At
the height $b$, again, only Betti zero changes, going from one to two, and we
also count the corresponding vertex as positive for~$\abc_0(K)$. At the height
of $c$, no Betti number changes, and thus there is no event in $\bc(K)$.
However, in a corresponding index filtration, the inclusion of the vertex at
height~$c$ increases $\beta_0$ by one, and the inclusion of the edge
then reduces $\beta_0$ by one. This is recorded in $\abc(K)$ at $c$ as an
additional positive simplex (going from two total to three total), and our
first negative simplex.  We see similar behavior in dimension zero at the
height of $d$. At height $d$ we also see $\beta_1$ go from zero to one,
which is recorded in the concise Betti function. In a corresponding index
filtration, the inclusion of the second edge at height $d$ increases
$\beta_1$, and is thus recorded as positive. 

Finally, we describe what happens in $\ecc(K)$ and $\aecc(K)$. At both heights
$a$ and $b$, the Euler characteristic of the sublevel sets increases by
one. Since this is due to inclusions of vertices, which are even-dimensional
simplices, both of these increases are recorded as even-dimensional in
$\aecc(K)$. At $c$, the Euler characteristic remains the same, so no change
occurs in $\ecc(K)$. In a corresponding index filtration, we see the vertex (an
even-dimensional simplex) and an edge (odd-dimensional simplices), which are
recorded in $\aecc(K)$. Finally, the Euler characteristic at
$d$ changes from two to zero, which is recorded directly in $\ecc(K)$.
An index filtration witnesses the appearance of one even and two odd
simplices at height $d$, and this is recorded in $\aecc(K)$.

\begin{figure*}
    \centering
    \begin{subfigure}[b]{0.3\textwidth}
        \includegraphics[scale=.6]{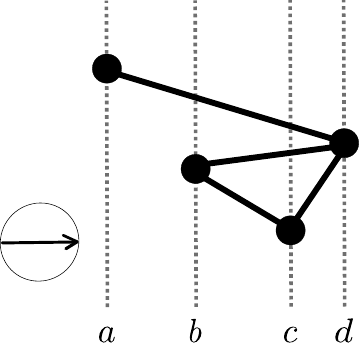}
        \subcaption{$\simComp$}
        \label{subfig:six-K}
    \end{subfigure}
    \begin{subfigure}[b]{0.3\textwidth}
        \includegraphics[scale=.6]{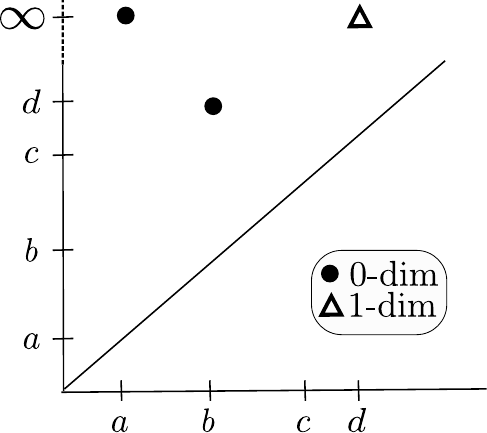}
        \subcaption{$\pd(K)$}
        \label{subfig:six-pd}
    \end{subfigure}
    \begin{subfigure}[b]{0.3\textwidth}
        \includegraphics[scale=.6]{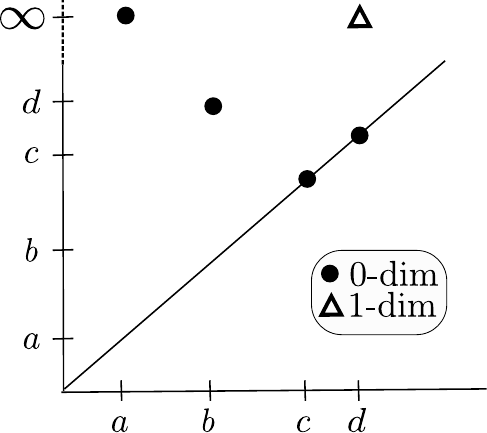}
        \subcaption{$\apd(K)$}
        \label{subfig:six-apd}
    \end{subfigure}
    \begin{subfigure}[b]{0.3\textwidth}
        \includegraphics[scale=.6]{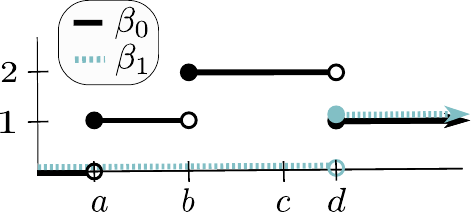}
        \subcaption{$\bc(K)$}
        \label{subfig:six-bc}
    \end{subfigure}
    \begin{subfigure}[b]{0.6\textwidth}
        \includegraphics[scale=.6]{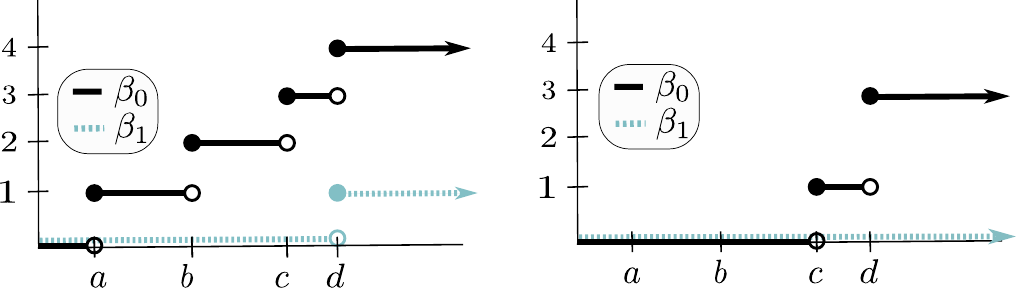}
        \subcaption{$\abc(K)$ (left: positive for ${\beta}_k$, right:
        negative for ${\beta}_k$}
        \label{subfig:six-abc}
    \end{subfigure}
    \begin{subfigure}[b]{0.3\textwidth}
        \includegraphics[scale=.6]{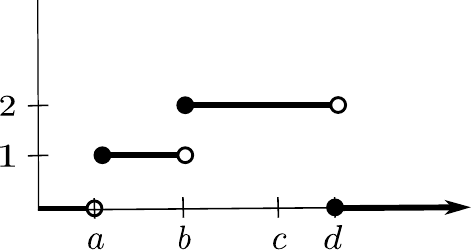}
        \subcaption{$\ecc(K)$}
        \label{subfig:six-ecc}
    \end{subfigure}
    \begin{subfigure}[b]{0.6\textwidth}
        \includegraphics[scale=.6]{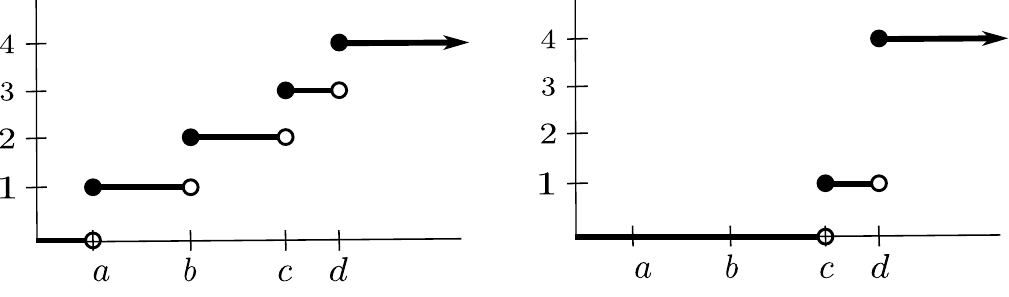}
        \subcaption{$\aecc(K)$ (left: even-dimensional, right: odd-dimensional)}
        \label{subfig:six-aecc}
    \end{subfigure}
    \caption{Six descriptors corresponding to the lower-star filtration in
    the direction indicated by the arrow of the simplicial complex in (a).
    }
    \label{fig:six}
\end{figure*}

\section{Zoo of Other Descriptor Types}
\label{append:other-descriptors}
In \secref{preliminary}, we adopt a general definition of topological
descriptor (\defref{descriptors}). In this appendix, we explore non-standard
topological descriptors and corresponding scenarios that may arise as a result
of this generality. The descriptors presented here are not intended to be taken
as anything that would necessarily make sense to use in practice, but rather,
as a sort of zoo of examples to get a quick glance at the mathematical extremes
and properties of the space of strength classes of topological descriptors.

First, we give an example of two distinct descriptor types that are in the same
equivalence class of strength.
\begin{example}
    \label{ex:sameclass}
Consider the topological descriptor denoted $-\pd$ that takes a lower-star
filtration in direction $s$, and produces~$\pd(-s)$, the persistence diagram in
direction $-s$. Although generally~$\pd(s)
\neq \pd(-s)$ (as multisets), a faithful
set~$\descSet{\pd}{K}{S}$ has the same cardinality as the faithful
set~$\descSet{-\pd}{K}{-S}$, and if a simplicial complex has no faithful set of
type~$\pd$, then it has no faithful set of type~$-\pd$.  Thus,~$\st{\pd} =
\st{-\pd}$.
\end{example}

Next, we observe that many examples of topological descriptors are not capable
of faithfully representing most simplicial complexes, such as the following.

\begin{descriptor}[First Vertex]
    \label{desc:convexcloud}
    Consider a descriptor $D_{V}$ that returns (1) the coordinates of the
    first vertex (or vertices) encountered and (2) the cardinality of the vertex set, but
    no other information.
\end{descriptor}

If the filtrations are directional filtrations, then
this descriptor is only faithful for convex point clouds.
Any set of vertices that defines the corners of a convex region can be
faithfully represented by this \descref{convexcloud}. However, since no vertex
interior to the convex hull nor any higher dimensional simplices are witnessed
by any direction, this descriptor type is incapable of faithfully representing
any other type of simplicial complex.

We can also construct descriptor types that are simply never able to form
faithful sets.

\begin{descriptor}[Trivial]
    \label{desc:incapable}
    Consider the trivial descriptor type $D_{0}$ that returns zero for all sublevel sets in
a filtration.
\end{descriptor}

Although this trivial descriptor type is an invariant of any filtration, it can
not faithfully represent any simplicial complex.
Thus, $\Gamma(K,D_0)=\aleph_{\top}$ for all $K$.
And so, in the space of all topological descriptors, \descref{incapable}
is in the minimum strength class.  We can (also trivially)
construct a descriptor type that is in the maximum strength class.

\begin{descriptor}[Filtration-Returning]
    \label{desc:incapable}
    Consider the descriptor type $D_{Filt}$ that returns the input~filtration.
\end{descriptor}

Thus, a single descriptor of this type is always
	faithful for a simplicial complex.

Finally, we can find instances of topological descriptors that are able to
faithfully represent a simplicial complex, but with a set no smaller than
uncountably infinite.

\begin{descriptor}[Indicator]\label{desc:infinite}
    Let~$K$ be a simplicial complex immersed in~$\R^d$ let $D_{\R}$ be a descriptor type
    parameterized by $\R^d$ that is constant over
    a filtration and is defined by
    \[
        D(K, x) = \begin{cases}
            1 \,\,\,\, \text{if   } x \in |K| \\
            0 \,\,\, \text{else}.
        \end{cases}
    \]
\end{descriptor}

Note that then, $\descSet{D}{K}{\R^d}$ is the (only) minimum faithful set for
$K$, and so $\Gamma(K,D_{\R})=\aleph_1$ for all $K$.
Thus, the (minimal) strength class of \descref{infinite} is
greater than the strength class of the trivial descriptor in \descref{incapable},
and there are no strength equivalence classes between them.

We now know the space of strength classes of topological descriptors has a
minimum and maximum, and we have identified a second smallest descriptor type;
is it a total order?  The following example shows that it
is not; there do indeed exist incomparable
descriptor~types.

\begin{lemma}[Incomparable Strength Classes]\label{lem:incomparable}
    There exist incomparable strength classes of topological descriptor types.
\end{lemma}

\begin{proof}
    Let $\gen_V$ denote \descref{convexcloud}. That is,
    given a direction $s$, $\gen_V$ returns: (1) the coordinates of the
    lowest vertex (or vertices) in direction~$s$, and (2) the cardinality of the
    vertex set.  We compare~$\gen_V$ with
    verbose persistence diagrams.
    Let~$v_1 = (0,0)$ and~$v_2 = (0,1)$.

    First, consider the simplicial complex $\simComp = \{v_1\}$. Then, regardless
    of direction, a single descriptor of is faithful for $\simComp$.
    However, since $K$ is in $\R^2$, any faithful set of \apdnames must have
    at least two linearly independent directions to recover both
    coordinates of~$\simComp$.

    Next, consider the simplicial complex $\simComp' = \{v_1, v_2\}$.
    No set of descriptors of type $\gen_V$ is faithful for~$\simComp'$
    (it cannot distinguish between $\simComp'$ and the simplicial complex consisting
    of the two disconnected vertices~$v_1$ and~$v_2$ without an edge). However,
    two \apdnames
    suffice to form a faithful set for $\simComp'$; for example, using the
    standard basis vectors~$\{e_1,e_2\}$ as the
    set. Thus, if
    $\descSet{\gen_V}{K}{S_{\gen_V}}$ and $\descSet{\apd}{K}{S_{\apd}}$
    are both minimal faithful sets, we see
    that~$|S_{\gen_V}| < |S_{\apd}|$ but $|S_{\gen_V}| >
    |S_{\apd}|$. Thus, although we have shown $\st{\gen_V} \neq
    \st{\apd}$, they are
    incomparable.
\end{proof}

Finally, we give a lemma that shows an impact of not assuming general position
of vertices.  

\begin{lemma}[Concise Equality]\label{lem:trivial}
    Without \assref{mostgeneral}, the strength equivalence classes of
    the three concise topological descriptor types from \secref{descriptors} are
    all~equal.
\end{lemma}

\begin{proof}
        Let $D \in \{ \ecc, \bc, \pd\}$.
	We must consider faithful sets of such descriptors for an arbitrary
        simplicial complex $K$ immersed in $\R^d$ (that may not be
        in general position). The argument differs depending on if $K$
	is a vertex set, or contains at least one edge; we consider each
	case.

        First, suppose~$n_1 =0$.  Then, $K$ has no edges and is a vertex set,
        meaning each vertex is a maximal simplex of $K$. Then, by \corref{perp},
        faithful sets of type~$D$ must include descriptors from at least $d+1$
        directions.
        By \lemref{envelopes}, the envelopes of each vertex must be
        zero-dimensional. Since the only zero-dimensional convex sets are
        singleton points, the envelope of each vertex contains
        that vertex and nothing else.

        Let $S$ be a set of $d+1$ directions such that the envelope of each
        vertex is zero-dimensional (note that such a set exists, for example,
        the standard basis directions ($e_i$), and the negative diagonal
        direction, $-1/\sqrt{d} (1, 1, \ldots, 1)$).
        Consider a single direction, $s \in S$ and $a \leq b \in \R^d$.  If no
        event occurs in $\gen(s)$ between heights $a$ and $b$, we know no
        connected component of~$K$ has
        its lowest vertex (or vertices) with respect to $s$ in the range from
        $a$ to $b$.  Then, from $\gen(s)$, we identify that $K$ has $n_0$
        connected components, and we know their starting heights with respect
        to the $s$ direction.  In other words, we know these connected
        components are contained in the (closed) upper half-spaces defined by $s$.

        Each additional direction in $S$ gives us more information about these
        $n_0$ connected components, and, just like~$s$, each additional
        direction provides an additional upper half-space in which we know the
        connected component is contained; each connected component must lie in
        the intersection of these half-spaces. Since we use a total of $d+1$
        directions, chosen so that the envelope of each vertex, $\E^S_v$ (the
        intersection of half-spaces), is zero-dimensional, we conclude each
        connected component is zero-dimensional.

        That is, we know the exact location of each vertex by identifying its
        envelope. Thus, minimal faithful sets of type~$D$ have cardinality of
        exactly~$d+1$, meaning the infimums considered in \defref{equal}
        are~$\Gamma(K,D) = d+1$ for such $K$.

        Next, suppose~$n_1 >1$.  We show
	no set of descriptors of type $D$ can faithfully represent~$K$.
        Let~$\tau$ be an edge in $K$ and construct another complex $L$ by starting
        with $K$ and taking the barycentric subdivision of~$\tau$ and all
        simplices containing $\tau$.  Then, since~$|K|=|L|$ the Euler
        characteristics/Betti numbers/homology throughout the filtrations of $K$
        or $L$ agree. That is, for every direction~$s$,~$\descSet{\gen}{K}{s} =
        \descSet{\gen}{L}{s}$. Since this is true for every $s$, the descriptor
        type $D$ is incapable of forming a faithful set for $K$, meaning the
        infimums considered in \defref{equal} are~$\Gamma(K, D) =
        \aleph_{\top}$ for such~$K$.

        We have shown that, for each $K$ without a general position assumption,
        we have $\Gamma(K,\ecc) = \Gamma(K, \bc) = \Gamma(K, \pd)$. Thus, when
        removing the general position assumption, we find
        $\st{\ecc} = \st{\bc} = \st{\pd}$.
\end{proof}

Fortunately, as shown in \secref{ordering}, the relations among concise
descriptors becomes more interesting when we assume general position. This also
has the benefit of reflecting the general position assumptions that are often
taken in practical applications.

\end{document}